\documentclass[runningheads]{llncs}
\usepackage{hyperref}
\usepackage{graphicx}
\usepackage[font=small,labelfont=bf]{caption}
\usepackage[font=small,labelfont=small]{subcaption}
\usepackage{amsmath}
\usepackage{amssymb}
\usepackage{cite}
\usepackage[capitalize]{cleveref}
\usepackage{enumitem}
\usepackage{xcolor}
\usepackage{todonotes}

\title{The Parameterized Complexity of Computing the Linear Vertex
  Arboricity\thanks{This is a shortened version of the master's thesis
    of the first author~\cite{erhardt-MTh21}.}}

\titlerunning{The Parameterized Complexity of Computing the Linear
  Vertex Arboricity}

\author{Alexander Erhardt\inst{1} \and Alexander~Wolff\inst{1}\orcidID{0000-0001-5872-718X}}

\authorrunning{A.~Erhardt and A.~Wolff}

\institute{Universit\"at W\"urzburg, W\"urzburg, Germany}

\definecolor{dark blue}{rgb}{0.121,0.47,0.705}
\let\emph\relax\DeclareTextFontCommand{\emph}{\color{dark blue}\em}
\renewcommand{\orcidID}[1]{\href{https://orcid.org/#1}{\includegraphics[scale=.03]{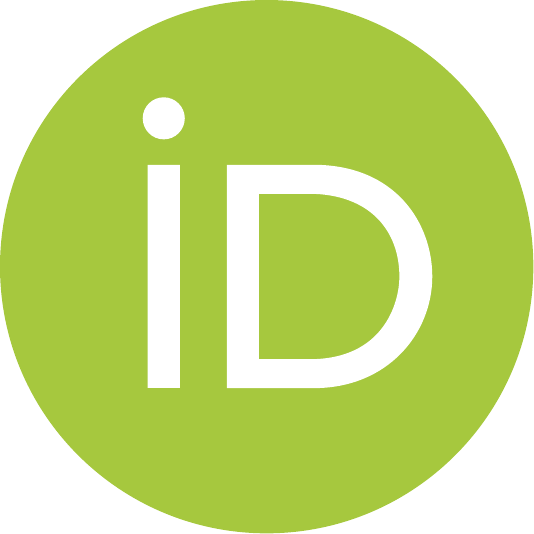}}}
\let\doendproof\endproof
\renewcommand\endproof{~\hfill$\qed$\doendproof}

\DeclareMathOperator{\seg}{seg}
\DeclareMathOperator{\lin}{line}
\DeclareMathOperator{\lva}{lva}

\begin{document}
	
\maketitle
	
\begin{abstract}
  The \emph{linear vertex arboricity} of a graph is the smallest
  number of sets into which the vertices of a graph can be partitioned
  so that each of these sets induces a linear forest.  Chaplick et
  al.\ [JoCG 2020] showed that, somewhat surprisingly, the linear
  vertex arboricity of a graph is the same as the \emph{3D weak line
    cover number} of the graph, that is, the minimum number of
  straight lines necessary to cover the vertices of a crossing-free
  straight-line drawing of the graph in $\mathbb{R}^3$.  Chaplick et
  al.\ [JGAA 2023] showed that deciding whether a given graph has
  linear vertex arboricity~2 is NP-hard.

  In this paper, we investigate the parameterized complexity of
  computing the linear vertex arboricity.  We show that the problem is
  para-NP-hard with respect to the parameter maximum degree.  Our
  result is tight in the following sense.  All graphs of maximum
  degree~4 (except for~$K_4$) have linear vertex arboricity at most~2,
  whereas we show that it is NP-hard to decide, given a graph of
  maximum degree 5, whether its linear vertex arboricity is~2.
  Moreover, we show that, for {\em planar} graphs, the same question
  is NP-hard for graphs of maximum degree~6, leaving open the
  maximum-degree-5 case.  Finally, we prove that, for any $k \ge 1$,
  deciding whether the linear vertex arboricity of a graph is~at
  most~$k$ is fixed-parameter tractable with respect to the treewidth
  of the given graph.
  
  \keywords{Visual complexity \and weak line cover number \and linear
    vertex arboricity \and parameterized complexity}
\end{abstract}

\section{Introduction}

Various measures for the \emph{visual complexity} of the drawing of a
graph have been suggested.  For example, the \emph{strong line cover
  number} is defined to be, for a given crossing-free straight-line
drawing of a graph, the cardinality of the smallest set of straight
lines whose union contains (all vertices and edges of) the drawing.
Accordingly, the \emph{weak line cover number} is the cardinality of
the smallest set of straight lines whose union contains all vertices
of the drawing.  These and other measures for the visual complexity of
a drawing of a graph lead to measures for the visual complexity of a
{\em graph} by taking the minimum over all valid drawings.

In~$\mathbb{R}^2$, both versions of the line cover number are
restricted to planar graphs.  Chaplick et
al.~\cite{chaplick_etal:CompGeo20} showed that, in any space
$\mathbb{R}^d$ with $d>3$, the numbers are the same as the
corresponding numbers in~$\mathbb{R}^3$.  They also showed that the 3D
weak line cover number of a graph~$G$ is the same as the \emph{linear
  vertex arboricity} of~$G$, which is defined purely combinatorially,
namely as the smallest size of a partition of the vertex set of~$G$
such that each set induces a linear forest, that is, a collection of
paths.  Let \textsc{$k$-Lva} denote the decision problem of deciding
whether the linear vertex arboricity of a given graph is at most~$k$.
Matsumoto~\cite{Matsumoto90} provided an upper bound for the linear
vertex arboricity $\lva(G)$ of a connected nonempty graph~$G$, namely
$\lva(G) \le 1 + \lfloor \Delta(G)/2 \rfloor$, where $\Delta(G)$
denotes the maximum degree of $G$.  Moreover, if $\Delta(G)$ is even,
then $\lva(G) = 1 + \lfloor \Delta(G)/2 \rfloor$ if and only if $G$ is
either a cycle or a complete graph.
By applying a meta-theorem of Farrugia~\cite{f-vpfai-EJC04}, Chaplick
et al.~\cite{cflrvw-cdgfl-JGAA23} showed that \textsc{2-Lva} is
NP-hard.

\paragraph{Our Contribution.}

In this paper, we investigate the parameterized complexity of
\textsc{$k$-Lva}.  We first show that the problem is para-NP-hard with
respect to the maximum degree.  More
specifically, we show that \textsc{2-Lva} is NP-hard even if the
graph has maximum vertex degree~5; see \cref{sec:hardness-deg5}.  On
the other hand, due to Matsumoto's result~\cite{Matsumoto90} mentioned
above, \textsc{2-Lva} has a simple solution for graphs of degree at
most~4: among those, $K_5$ is the only no-instance.  In other words,
our hardness result is tight with respect to maximum degree.

We further prove that, restricted to planar graphs, \textsc{2-Lva} is
NP-hard for graphs of maximum degree~6; see
\cref{sec:hardness-planar}.  Annoyingly, this leaves the complexity of
\textsc{2-Lva} open for the class of planar graphs of maximum
degree~5.  Finally, we turn to \textsc{$k$-Lva}, that is, to the
problem of deciding whether the linear vertex arboricity of a given
graph is at most~$k$.  We show that \textsc{$k$-Lva} is
fixed-parameter tractable (FPT) with respect to the treewidth of the
given graph; see \cref{sec:fpt}.  Finally, we give compact integer
linear programming (ILP) and satisfiability (SAT) formulations for
\textsc{$k$-Lva} and close with some open problems; see
\cref{sec:ilp,sec:open}.

\paragraph{Related Work.}

Chaplick et al.~\cite{cflrvw-cdgfl-JGAA23} showed that computing the
2D strong line cover number $\lin(G)$ of a planar graph~$G$ is FPT
with respect to the natural parameter.  They observed that, for a
given graph~$G$ and an integer~$k$, the statement $\lin(G)\le k$ can
be expressed by a first-order formula about the reals.  This
observation shows that the problem of deciding whether or not
$\lin(G) \le k$ lies in $\exists\mathbb{R}$: it reduces in polynomial
time to the decision problem for the existential theory of the reals.
The algorithm of Chaplick et al.\ crucially uses the exponential-time
decision procedure for the existential theory of the reals by
Renegar~\cite{Renegar92a,Renegar92b,Renegar92c}.  Unfortunately, this
procedure does {\em not} yield a geometric realization.

Firman et al.~\cite{frw-wlcn-EuroCG18} showed that the Platonic graphs
all have 3D~weak line cover number~2 and that there is an infinite
family of polyhedral graphs with maximum degree~6, treewidth~3, and
unbounded 2D weak line cover number.  Biedl et
al.~\cite{bfmw-lpcnr-GD19} proved that the 2D~weak line cover number
of the universal stacked triangulation of depth~$d$ is $d+1$, that
this number is NP-hard to compute for general planar graphs, and that
every graph with $n$ vertices and 3D~weak line cover number~2 has at
most $5n-19$ edges.  Felsner~\cite{felsner-JoCG20} showed that every
plane graph with $n$ vertices and without separating triangle has
2D~weak line cover number at most $\sqrt{2n}$, whereas
Eppstein~\cite{eppstein-JoCG21} constructed, for every positive
integer~$\ell$, a cubic 3-connected planar bipartite graph with
$O(\ell^3)$ vertices and 2D~weak line cover number greater
than~$\ell$.

A \emph{segment} in a straight-line drawing of a graph is a maximal
set of edges that together form a line segment.  The \emph{segment
  number} $\seg(G)$ of a planar graph~$G$ is the minimum number of
segments in any planar straight-line drawing
of~$G$~\cite{dujmovic_etal:compGeo07}.  It is
$\exists \mathbb{R}$-complete \cite{orw-ysng-GD19}
(and hence NP-hard) to compute the segment number of a planar graph.
Clearly, $\lin(G) \le \seg(G)$
for any graph~$G$~\cite{chaplick_etal:CompGeo20}.  Moreover, Cornelsen
et al.~\cite{cdggkw-pcsn-GD23} proved that, for any connected
graph~$G$, it holds that $\seg(G) \le \lin^2(G)$.  They showed also
that computing the segment number of a graph is FPT with respect to
each of the following parameters: the natural parameter, the 2D strong
line cover number, and the \emph{vertex cover number}.  Recall that
the vertex cover number is the minimum number of vertices that have to
be removed such that the remaining graph is an independent set.

For circular-arc drawings of planar graphs, the \emph{arc
  number}~\cite{s-dgfa-JGAA15} and \emph{circle cover
  number}~\cite{krw-dgfcf-JGAA19} are defined analogously as the
segment number and the 2D strong line cover number, respectively, for
straight-line drawings.  For an example, see \cref{fig:numbers}.  All
these numbers have been considered as meaningful measures of the
visual complexity of a drawing of a graph; in particular, for the
segment number, Kindermann et al.~\cite{kindermannMS:JGAA2012} have
conducted a user study to understand whether drawings with low segment
number are actually better from a cognitive point of view.

\begin{figure}[tb]
  \centering
  \includegraphics{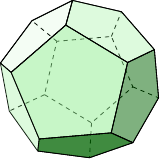} \hfil
  \includegraphics{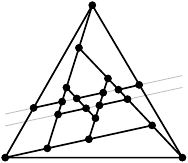} \hfil
  \includegraphics{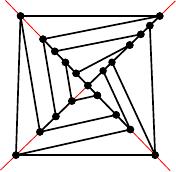} \hfil
  \includegraphics{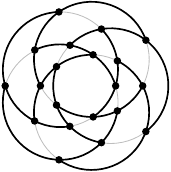}
  \caption{The dodecahedron graph has 20 vertices and 30 edges~(a),
    segment number~13, 2D strong line cover number~10~(b), 2D and 3D
    weak line cover number~2~\cite{frw-wlcn-EuroCG18}~(c), arc
    number~10, and circle cover number~5~\cite{krw-dgfcf-JGAA19}~(d).}
  \label{fig:numbers}
\end{figure}

\section{NP-Hardness of 2-LVA for Planar Graphs}
\label{sec:hardness-planar}

We reduce from \textsc{Clause-Linked-Planar-Exactly-3-Bounded-3-SAT},
a satisfiability problem with several additional restrictions
concerning the given 3-SAT formula~$\varphi$:
\begin{enumerate}[label=(F\arabic*),leftmargin=*]
\item The variable--clause incidence graph~$H_\varphi$ of~$\varphi$ is
  planar and admits a planar embedding with a linear arrangement of
  the clause vertices such that every two consecutive clause vertices
  can be connected by an extra edge without introducing any
  crossings.\label{enum:graph}
\item Every variable in~$\varphi$ occurs in exactly three clauses,
  once as a negated and twice as a positive
  literal.\label{enum:variable}
\item Every clause of~$\varphi$ consists of two or three variables.
  If a clause consists of three variables, none of them is negated in
  that clause.\label{enum:clause}
\end{enumerate}
The problem \textsc{Clause-Linked-Planar-Exactly-3-Bounded-3-SAT}~is
NP-hard~\cite{Fellows95}.  We \mbox{transform}~$(\varphi,H_\varphi)$
into a graph~$G$ of maximum degree~6 such that $\varphi$ has a
satisfiable truth assignment if and only if $\lva(G)=2$.

We first construct a basic building block $B$ that we use in the
variable gadget, in the clause gadget, and for the links between
consecutive clause gadgets.  The block $B$ is a triangulated planar
graph with seven vertices; see \cref{fig:BasicMaxDeg6}.  In the
following figures, we indicate the partition of the vertex sets of our
gadgets by using the colors white and gray.  A \emph{legal coloring}
of the vertices then corresponds to a partition into two sets each of
which induces a linear forest.  In such a legal coloring two simple
rules apply:
\begin{enumerate}
\item every cycle (in particular, every triangle) must be bichromatic
  (to avoid that one of the color classes induces a cycle), and
\item the neighborhood of every vertex of degree at least~3 must be
  bichromatic (to avoid that one of the color classes induces a star).
\end{enumerate}

\begin{figure}[tbh]
  \centering
  \includegraphics{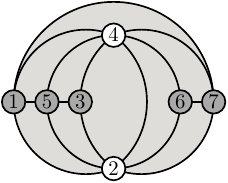}
  \caption{Basic building block $B$} 
  \label{fig:BasicMaxDeg6}
\end{figure}

\begin{lemma}
  \label{lem:basic-block}
  Assuming that vertex~1 of $B$ is gray, the unique legal coloring
  of~$B$ is as in \cref{fig:BasicMaxDeg6}. 
\end{lemma}

\begin{proof}
  Clearly the coloring in \cref{fig:BasicMaxDeg6} is legal.
  
  Our first claim is that the triangle 234 contains exactly one gray
  vertex.

  To this end, note that vertices~3 and~4 cannot both be gray.
  Indeed, this would force vertices~2 (due to triangle 234) as well as
  vertices~6 and~7 (due to the two gray neighbors~1 and~3 of~4) to be
  white, a contradiction since 267 is a triangle.

  For symmetry, vertices~2 and~3 cannot both be gray.
  Due to triangle 124, vertices~2 and~4 cannot both be gray.
  This settles our first claim.

  Our second claim is that the only gray vertex in triangle 234 is
  vertex~3.  Due to symmetry, it suffices to show that vertex~2 cannot
  be the gray vertex.  Indeed, this would force vertex~5 to be white
  (to avoid a gray triangle 125), but then triangle 345 would be
  white.

  This settles the second claim.  In other words, vertex~3 must be
  gray and vertices~2 and~4 must be white.  This implies that
  vertices~5, 6, and~7 must be gray (due to triangles 245, 246, and
  247, respectively).  This completes our proof.
\end{proof}

Now it is easy to extend~$B$ to a variable gadget~$V_a$ for the
variable~$a$; see \cref{fig:VariableMaxDeg6}.  We simply attach two
white vertices (labeled $a$) to vertex~1 and two white vertices
labeled $a'$ and $a''$ to vertex~7.  The latter two vertices are
adjacent to each other and a gray vertex~$\bar{a}$.

\begin{figure}[tbh]
  \centering
  \includegraphics{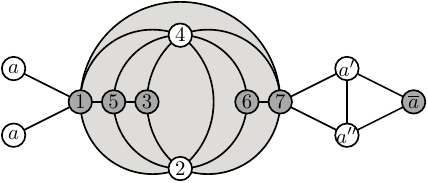}
  \caption{Variable gadget $V_{a}$ for variable $a$ in $\varphi$} 
  \label{fig:VariableMaxDeg6}
\end{figure}

\begin{lemma}
  \label{lem:variable-gadget}
  Let~$a$ be a variable in~$\varphi$.  Assuming that vertex~1 is gray,
  the unique legal coloring of~$V_a$ is as in \cref{fig:BasicMaxDeg6}.
  In particular, the two vertices labeled~$a$ have the same color,
  whereas the unique vertex labeled~$\bar{a}$ has a different color.
  Moreover, each of the three has a color different from its neighbors
  in~$V_a$.
\end{lemma}

\begin{proof}
  Due to \cref{lem:basic-block}, the copy of~$B$ has a unique legal
  coloring assuming that vertex~1 is gray.
  Observe that vertices~1 and~7 already have two gray neighbors
  within~$B$.  Hence their new neighbors (labeled $a$, $a'$, and
  $a''$) must all be white.  Since $a'$, $a''$, and $\bar{a}$ form a
  triangle, $\bar{a}$ must be gray.
\end{proof}

We now turn to the clause gadget, which consists of a 4- or 5-cycle
that goes around a copy of the basic building~$B$; see
\cref{fig:ClauseMaxDeg6}.  Two vertices of this outer cycle are
adjacent to vertices~1 and~7 of~$B$.  These are labeled~0; the other
two or three vertices of the outer cycle are labeled with the names of
the two or three literals that occur in the clause.  Later we will
identify these vertices with the vertices of the variable gadgets
labeled with the same names (see \cref{fig:VariableMaxDeg6}) and we
will synchronize the colors of the 0-labeled vertices of {\em all}
clause gadgets.

\begin{figure}[tbh] 
  \centering
  \includegraphics[page=1]{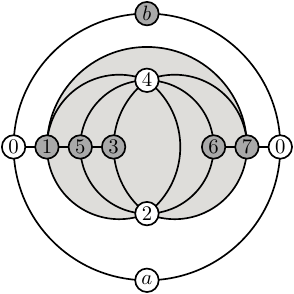} \hfil
  \includegraphics[page=2]{figures/Klausel-Gadget-MD6}  
  \caption{The gadgets for a two-variable clause $a \lor b$ (left) and
    for a three-variable clause $a \lor b \lor c$ (right).  Both use
    copies of the basic building block~$B$.}
  \label{fig:ClauseMaxDeg6}
\end{figure}

\begin{lemma}
  \label{lem:clause-gadget}
  Assuming that the left 0-labeled vertex of the clause gadget is
  colored white, in any unique legal coloring the right 0-labeled
  vertex is also white.  The coloring of the vertices labeled with the
  names of literals is arbitrary except they cannot all be white; see
  \cref{fig:ClauseMaxDeg6}.
\end{lemma}

\begin{proof}
  Assuming that the left 0-labeled vertex of the clause gadget is
  white, vertex~1 of the copy of~$B$ must be gray since it has two
  neighbors of the same color (vertices~5 and~7) due to
  \cref{lem:basic-block}.  This fixes the coloring of the copy of~$B$
  and forces the right 0-labeled vertex to be white.  The remaining
  two or three vertices labeled with the names of literals can be
  colored arbitrarily as long as they are not all colored white (the
  color of the two 0-labeled vertices) because then the set of white
  vertices would induce a cycle.
\end{proof}

Obviously, \textsc{2-Lva} is in NP since we can easily verify a given
solution in polynomial time.  Next we show the NP-hardness.

\begin{theorem}
  \label{thm:hardness-planar}
  \textsc{2-Lva} is NP-hard even for planar graphs of maximum
  degree~6.
\end{theorem}

\begin{proof}
  We now describe our complete reduction.  Given an
  instance~$(\varphi,H_\varphi)$ of
  \textsc{Clause-Linked-Planar-Exactly-3-Bounded-3-SAT}, we construct
  a planar graph~$G$ using the planar embedding of~$H_\varphi$ as a
  pattern; see \cref{fig:Complete-Example-MD6}.  For each variable~$a$
  in~$\varphi$, $G$ contains the variable gadget~$V_a$, for each
  clause $a \lor b$ or $a \lor b \lor c$ (where $a$, $b$, and $c$ are
  literals), $G$ contains the corresponding clause gadget.  The two or
  three vertices in the clause gadget that are labeled with the names
  of literals are identified with the vertices that are labeled with
  the same names and that are part of the corresponding variable
  gadgets.  Due to property~\ref{enum:graph}, we can think of the
  clause vertices in~$H_\varphi$ to be connected by a path of extra
  edges.  Following this path, we connect the corresponding pairs of
  clause gadgets in~$G$ by copies of the basic building block~$B$.
  Specifically, we connect vertex~1 (and vertex~7) of~$B$ to the right
  (left) 0-labeled vertex of the left (right) neighboring clause
  gadget.  This completes the description of~$G$.

  It is easy to see that the reduction can be performed in polynomial
  time and that the maximum degree of~$G$ is~6.  Also note that if the
  left 0-labeled vertex in the leftmost clause gadget is colored white
  then, according to \cref{lem:clause-gadget}, the right 0-labeled
  vertex of that gadget is also white.  Arguing as in that proof, we
  get that the left 0-labeled vertex of the next clause gadget is also
  white etc.  In other words, all 0-labeled vertices must have the
  same color.

  To show the correctness of our reduction, we first assume that
  $(\varphi,H_\varphi)$ is a yes-instance, that is, there is a
  satisfing truth assignment for~$\varphi$.  We need to show that then
  $G$ admits a legal coloring.  For each variable $a$ in~$\varphi$
  that is set to true, we color the vertices labeled $a$ in~$V_a$ gray
  and the vertex labeled~$\bar a$ white.  If $a$ is set to false, we
  swap the two colors.  (This is the case in
  \cref{fig:Complete-Example-MD6}.)  We color all 0-labeled vertices
  white and all copies of~$B$ inside the clause gadgets and between
  neighboring clause gadgets as in \cref{fig:BasicMaxDeg6}.

  In either case, we can extend the coloring of~$V_a$ to a legal
  coloring.  This is due to \cref{lem:variable-gadget}.  Since the
  truth assignment is satisfiable, for each $i \in \{1,2,\dots,m\}$,
  at least one of the three literals in clause~$C_i$ is true, so the
  vertex with the corresponding label will be gray and not all
  vertices of the outer cycle of the gadget for~$C_i$ are white.
  Hence all outer cycles of the clause gadgets have a legal coloring
  together with the white 0-labeled vertices.  Every copy of~$B$ in
  the chain of clause gadgets can be colored as in
  \cref{fig:BasicMaxDeg6}.  Thus, the coloring of~$G$ is legal.

  In the reverse direction, assume that $G$ is a yes-instance of
  \textsc{2-Lva}.  Then all 0-labeled vertices have the same color,
  say, white.  We assign the value true to each variable~$a$ where the
  vertices with that label are gray.  To the other variables, we
  assign the value false.  Due to \cref{lem:variable-gadget}, the
  vertices with label~$\bar{a}$ have a different color than those with
  label~$a$.  Since, in each clause gadget, at least one of the
  vertices that are labeled with literal names must be gray, the
  corresponding literal evaluates to true and the clause is satisfied.
  Thus, the variable assignment is satisfying.
\end{proof}

\begin{figure}[tbh]
  \centering
  \includegraphics[width=\textwidth]{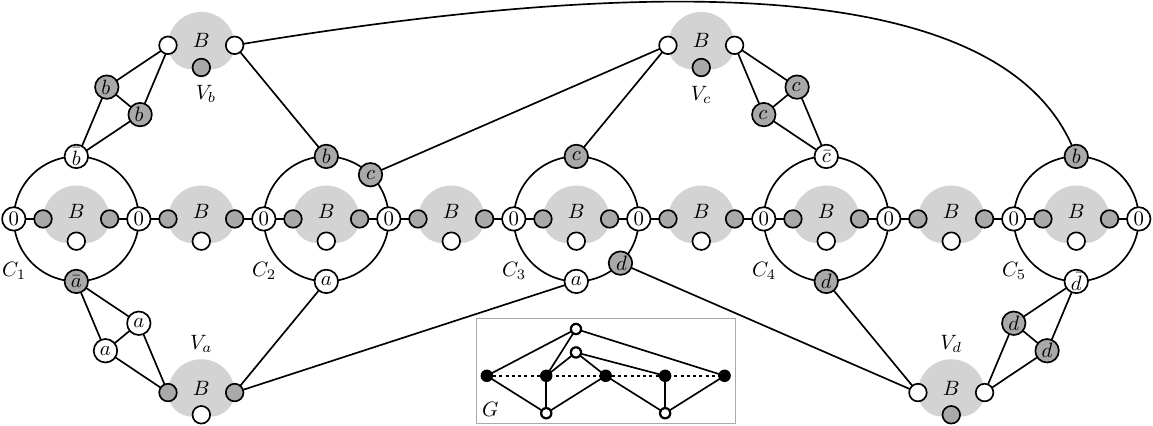}
  \caption{The graph $H_\varphi$ with black clause vertices and white variable
    vertices (see inset) and the resulting graph $G$ for
    $\varphi = (\bar{a} \lor \bar{b}) \land (a \lor b \lor c) \land (a
    \lor c \lor d) \land (\bar{c} \lor d) \land (b \lor \bar{d})$}
  \label{fig:Complete-Example-MD6}
\end{figure}

\section{NP-Hardness of 2-LVA for Graphs of Maximum Degree~5}
\label{sec:hardness-deg5}

For simplicity, we again reduce from
\textsc{Clause-Linked-Planar-Exactly-3-Bounded-3-SAT} although we do
not really need planarity here.  Our reduction is similar to that in
the proof of \cref{thm:hardness-planar} but we use gadgets with
maximum degree~5 instead of~6.

Here, our basic building block is $K_5^-$, that is, the graph that is
obtained from~$K_5$ by removing one edge; see \cref{fig:K5-}.  Due to
Matsumoto's result~\cite{Matsumoto90}, we have $\lva(K_5^-)=2$.  Note
that the two vertices of degree~3 (vertices~1 and~2 in \cref{fig:K5-})
must have the same color in a legal coloring, say, white.  Otherwise
adding the missing edge would not change the linear vertex arboricity
of~$K_5^-$, contradicting $\lva(K_5)=3$.  Observe that among the
remaining three vertices (of degree~4) exactly one must also be white.
Indeed, the three form a triangle, so they cannot all be gray.  On the
other hand, if two of the degree-4 vertices were white, they would form
a white triangle with, e.g., vertex~1.

\begin{figure}[tbh] 
  \begin{minipage}[b]{.42\textwidth}
    \centering
    \includegraphics{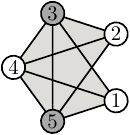}
    \caption{The basic building block $K_{5}^-$} 
    \label{fig:K5-}
  \end{minipage}
  \hfill
  \begin{minipage}[b]{.54\textwidth}
    \centering
    \includegraphics{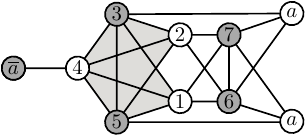}
    \caption{The variable gadget~$V_a'$ for variable~$a$} 
    \label{fig:VarK5}
  \end{minipage}
\end{figure}

Now we extend our basic building block to a variable gadget; see
\cref{fig:VarK5}.  First we add a pair of adjacent vertices (6 and~7)
and connect both of them to both vertices~1 and~2.  Then we add two
nonadjacent vertices labeled~$a$ and connect both of them to both
vertices~1 and~2.  Next we connect the top $a$-labeled vertex to
vertex~3 and the bottom one to vertex~5.  Finally we add a vertex
labeled~$\bar{a}$ and connect it to vertex~4.

\begin{lemma}
  \label{lem:variable-gadget-deg5}
  Assuming that one of the $a$-labeled vertex is white, the unique
  legal coloring of the variable gadget $V_a'$ is as in
  \cref{fig:VarK5}.  In particular, the two vertices labeled~$a$ have
  the same color, whereas the unique vertex labeled~$\bar{a}$ has a
  different color.  Moreover, each of the three has a color different
  from its neighbors in~$V_a'$.
\end{lemma}

\begin{proof}
  Let's assume that the upper vertex labeled~$a$ is white; the other
  case is symmetric.  This implies that the neighbors of that vertex
  (i.e., vertices~3, 6, and~7) must all be gray.  Since the lower
  $a$-vertex forms a triangle with vertices~6 and 7, it must be white,
  which forces vertex~5 to be gray.  Now it is clear that vertex~4
  must be white and vertex~$\bar a$ must be gray.
\end{proof}

Next we set up the clause gadget; see \cref{fig:clause-gadget-deg5}.
For three-variable (two-variable) clause, the gadget consists only of
a 6-cycle (5-cycle).  Three of the vertices are labeled~0; the others
are labeled with the names of the corresponding literals.

As in the proof of \cref{thm:hardness-planar}, we build a chain of
clause gadgets.  In this chain, we link two consecutive gadgets as
follows.  We connect the three 0-labeled vertices to those of the next
clause gadget via three vertices labeled~1, 2, and~3 that form a path
of length~2 with vertex~2 being in the middle.  Each of these three
vertices is connected to three 0-labeled vertices: vertex~2 has two
such neighbors on the left and one on the right; for vertices~1 and~3
it is opposite.

\begin{figure}[tbh] 
  \centering
  \includegraphics[page=4]{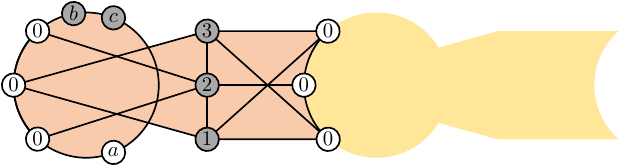}
  \caption{Gadget for the clause $a \lor b \lor c$}
  \label{fig:clause-gadget-deg5}
\end{figure}

\begin{lemma}
  \label{lem:clause-gadget-deg5}
  Assuming that the three 0-labeled vertices of a clause gadget have
  the same color, those of the next clause gadget in the chain must
  have the same color.
\end{lemma}

\begin{proof}
  Assuming that the 0-labeled vertices of the left clause gadget
  (light orange in \cref{fig:clause-gadget-deg5}) are all white, the
  middle 0-labeled vertex has two white neighbors (the other 0-labeled
  vertices), so its other neighbors (vertices~1 and~3) must be gray.
  Vertex~2 forms a cycle of length~4 together with the three 0-labeled
  vertices, so vertex~2 must also be gray.

  Similarly, vertex~2 has two gray neighbors (vertices~1 and~3), hence
  the middle 0-labeled vertex of the next clause gadget (indicated in
  yellow in \cref{fig:clause-gadget-deg5}) must be white.  The other
  two 0-labeled vertices of that gadget each form a triangle with
  vertices~1 and~3, so both of them must be white.
\end{proof}

The last ingredient for our construction is the \emph{starter gadget}.
It ensures that the 0-labeled vertices of the first clause gadget in
the chain must receive the same color.

\begin{figure}[tb]
  \centering
  \includegraphics[page=3]{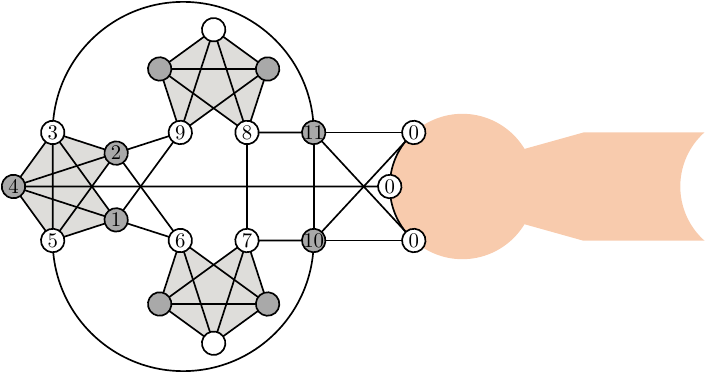}
  \caption{The starter gadget consists of three copies of $K_{5}^{-}$
    and a~$K_2$.  It ensures that the 0-labeled vertices of the first
    clause gadget in the chain must receive the same color.}
  \label{fig:starter-gadget-deg5}
\end{figure}

\begin{lemma}
  \label{lem:starter-gadget-deg5}
  Assuming that vertex~1 of the starter gadget is colored gray, the
  unique legal coloring of the starter gadget is as in
  \cref{fig:starter-gadget-deg5}.  In particular, the three 0-labeled
  vertices to which the starter gadget is connected all have the same
  color in a legal coloring.
\end{lemma}

\begin{proof}
  If vertex~1 is gray, vertex~2 must also be gray as we have observed
  in the beginning of this section.  Since one of the vertices~3, 4,
  and 5 must also be gray, vertices~6 and 9 must be white.  This
  forces vertices~7 and~8 to also be white.  Since~7 and~8 are
  adjacent and each have another white neighbor within their copies
  of~$K_5^-$, their neighbors~10 and~11 must be gray.  Now it is clear
  that vertices~3 and~5 must be white (otherwise they would have three
  neighbors in their own color), and hence vertex~4 is gray.  Since~4
  is gray and has two gray neighbors, the middle 0-labeled vertex must
  be white.  The fact that~10 and~11 are gray forces the other two
  0-labeled vertices, which both form triangles with each of~10
  and~11, to also be white.
\end{proof}

\begin{theorem}
  \textsc{2-Lva} is NP-hard, even for graphs of maximum degree~5.
\end{theorem}

\begin{proof}
  We now describe our reduction.  Given an instance
  $(\varphi,H_\varphi)$ of
  \textsc{Clause-Linked-Planar-Exactly-3-Bounded-3-SAT}, we again
  construct a planar graph~$H'$ using the planar embedding
  of~$H_\varphi$ as a pattern.  Variable gadgets interact with clause
  gadgets as in the proof of \cref{thm:hardness-planar}, clause
  gadgets again form a chain where consecutive gadgets are linked with
  copies of the three vertices~1, 2, and~3 in
  \cref{fig:clause-gadget-deg5}.  Structurally the only difference is
  that the chain starts with a single occurrence of the starter gadget
  shown in \cref{fig:starter-gadget-deg5}.  This completes the
  description of the graph~$H'$.  Clearly, the reduction takes
  polynomial time and the resulting graph~$H'$ has maximum vertex
  degree~5.
  
  Given Lemmas~\ref{lem:variable-gadget-deg5},
  \ref{lem:clause-gadget-deg5}, and~\ref{lem:starter-gadget-deg5}, the
  correctness of our reduction follows similarly as in the proof of
  \cref{thm:hardness-planar}.
\end{proof}

\section{An FPT Algorithm for $k$-LVA}
\label{sec:fpt}

In this section we show that \textsc{$k$-Lva} is fixed-parameter
tractable with respect to treewidth.  We use Monadic Second-Order
Logic (MSO$_2$) -- a subset of {\em second-order logic}~-- to express
for a graph~$G$ that $\lva(G) \le c$.

MSO$_2$ formulas are built from the following primitives.
\begin{itemize}
\item variables for vertices, edges, sets of vertices, and sets of edges;
\item binary relations for: equality ($=$), membership in a set ($\in$), subset of a set ($\subseteq$),
 and edge--vertex incidence ($I$);
\item standard propositional logic operators: $\lnot$, $\land$,
  $\lor$, $\rightarrow$, and $\leftrightarrow$.
\item standard quantifiers ($\forall,\exists$) which can be applied to
  all types of variables.
\end{itemize}
For a graph $G$ and an MSO$_2$ formula~$\phi$, we use $G \models \phi$
to indicate that $\phi$ can be satisfied by $G$ in the obvious way.
Properties expressed in this logic allow us to use the powerful
algorithmic result of Courcelle (and Engelfriet) stated next. 

\begin{theorem}[\!\!\cite{courcelle1990,courcelle2012}]
  \label{thm:courcelle}
  For any integer $t \geq 0$ and any MSO$_2$ formula $\phi$ of
  length~$\ell$, an algorithm can be constructed that takes a graph
  $G$ with $n$ vertices, $m$ edges, and treewidth at most $t$, and
  decides whether $G \models \phi$ in time $O(f(t,\ell)\cdot (n+m))$,
  where the function $f$ is computable.
\end{theorem}

We now show how Courcelle's result helps us.  Note that there is an
algorithm running in time $2^{O(t)}\cdot n$ that, given a graph~$G$
and a positive integer~$t$, either constructs a tree decomposition
of~$G$ of width at most $2t+1$ or reports that the treewidth of~$G$ is
greater than~$t$ \cite{korhonen}.

\begin{theorem}
  \label{thm:fpt}
  Given a graph~$G$ with $n$ vertices, $m$ edges, and treewidth at
  most~$t$, $\lva(G)$ can be computed in $O(f(t)\cdot(m+n))$ time,
  where $f$ is some computable function.  In other words,
  \textsc{$k$-Lva} is FPT with respect to treewidth.
\end{theorem}

\begin{proof}
  Let $\tau$ be the treewidth and let $\chi$ be the chromatic number
  of the given graph~$G$.  Since every independent set is a linear
  forest, we have $k \le \chi$.  On the other hand, it is well-known
  that $\chi \le \tau+1$.
  Hence, $k \le \tau+1$, and $\tau$ is the only remaining parameter.
  
  We formulate \textsc{$k$-Lva} in MSO$_2$.  In the following, we
  repeatedly use $V$ for $V(G)$, $E$ for $E(G)$, $U$ for a subset
  of~$V$, $u$ and $v$ for vertices of~$G$, $e$ for an edge of~$G$, and
  $I(e,v)$ 
  to express that vertex~$v$ is incident to edge~$e$.

  We first construct a predicate that expresses, for a set~$U$ of
  vertices, that $G[U]$ is connected.  To this end we require, for
  every nonempty proper subset $U'$ of~$V$, that there exists at least
  one edge that connects a vertex in $U'$ to a vertex
  in~$U \setminus U'$.
  \[\textsc{Conn($U$)} \equiv (\forall U' \colon \emptyset \ne U'
    \subsetneq U) (\exists v \in U')(\exists w \in U \setminus U')
    (\exists e \in E) [I(e,v) \land I(e,w)]\] %

  Next we set up a predicate that expresses that a set~$U$ of vertices
  induces a graph $G[U]$ that is a collection of cycles.  We simply
  require that for every vertex~$v$ in~$U$ has two neighbors in~$U$.
  (The notation $\exists^{=2} w \in U$ encodes the existence of
  exactly two vertices $w$ in~$U$.)
  \[\textsc{Cycle-Set($U$)} \equiv (\forall v \in U )(\exists^{=2} w
    \in U)(\exists e \in E)[I(e,v) \land I(e,w)]\]%
  Combining the last two predicates yields a predicate that
  expresses that $G[U]$ is a single cycle.
  \[\textsc{Cycle($U$)} \equiv \textsc{Cycle-Set($U$)}
    \land\textsc{Conn($U$)}\]%
  Next, we express that a set $U$ of vertices contains a \emph{star},
  that is, a vertex of degree at least~3 in~$G[U]$.
  $$\textsc{Star($U$)} \equiv (\exists v \in U )( \exists^{\ge 3} w
  \in U)(\exists e \in E)[I(e,v) \land I(e,w)]$$
  Now we express that a set $U$ of vertices is a linear forest, that
  is, no subset of $U$ induces a cycle or contains a star.
  $$\textsc{Path-Set($U$)} \equiv (\forall U' \subseteq U)[\neg
  \textsc{Cycle(}U'\textsc{)} \land \neg \textsc{Star(}U'\textsc{)}]$$

  To express a completely different graph property, Chaplick et
  al.~\cite{ckllw-bo-GD17} defined the predicate
  \textsc{Vertex-Partition($U_1, U_2, \dots, U_k$)} (VP for short),
  which is true if the sets $U_1, U_2, \dots, U_k$ form a partition of
  the vertex set~$V$ of~$G$.  In other words, every vertex $v$ of $G$
  must be contained in exactly one of the sets $U_1, U_2, \dots, U_k$.
  $$\mathrm{VP}(U_1, U_2, \dots, U_k) \equiv (\forall v) \Bigg[
  \Bigg( \bigvee_{i=1}^k v \in U_{i} \Bigg) \land \Bigg(
  \bigwedge_{i \neq j} \neg( v \in U_{i} \land v \in U_{j}) \Bigg)
  \Bigg]$$
  Finally, we can express \textsc{$k$-Lva} in MSO$_2$:
  $$\textsc{$k$-Lva($U$)} \equiv (\exists U_1, U_2,\dots, U_k)
  \Bigg[\mathrm{VP}(U_1, U_2, \dots, U_k) \land \bigwedge_{i=1}^{k}
  \textsc{Path-Set($U_{i}$)}\Bigg].$$
\end{proof}

\section{Compact ILP and SAT Formulations}
\label{sec:ilp}

In this section we give a compact ILP formulation that can easily be
turned into an equivalent SAT formulation using the same set of
variables.  (If an ILP formulation just encodes Boolean expressions,
the corresponding SAT formulation can usually be solved faster.)  At
first glance, it seems difficult to forbid cycles in a compact way,
but for \textsc{$k$-Lva} this can be accomplished by enforcing a
transitive ordering of the vertices, with the addition constraint that
vertices in the same set of the $k$-partition may only be adjacent in
the given graph if they are neighbors in the ordering.  This at the
same time ensures that every set in the partition is a linear forest.
(To see the other direction, given a partition into linear forests, we
can obtain such a special transitive ordering by concatenating the
paths of the linear forests arbitrarily.)

As in \cref{sec:hardness-planar}, we encode the sets in the desired
$k$-partition of the given graph~$G$ by coloring the vertices of~$G$
with a color in~$[k]$, which we use as shorthand for
$\{1,2,\dots,k\}$.  A coloring is legal if no color class induces a
cycle or a star (that is, a vertex of degree greater than~2).  We
introduce, for each vertex~$v$ of~$G$ and for each color $i \in [k]$,
a binary variable~$c_{v,i}$.  The intended meaning of $c_{v,i}=1$ is
that vertex~$v$ receives color~$i$.  To ensure that each vertex
receives a single color, we introduce the following constraints:
\begin{align}
  \label{eq:1}
  \sum_{i \in [k]} c_{v,i} &= 1 && \forall v \in V(G). \\
  \intertext{In a SAT formulation, this can be expressed by
  $\bigvee_{i \in [k]} c_{v,i}$ and, additionally,
  $\neg (c_{v,i} \land c_{v,j})$ for every pair
  $(i,j) \in [k]^2$ with $i \ne j$.
  Further, we introduce, for each pair of vertices~$u$
  and~$v$ of~$G$ with $u \ne v$, a binary variable~$x_{uv}$.  The
  intended meaning of $x_{uv}=1$ is that $u$ precedes~$v$ in the
  desired ordering.  We want the ordering to be antisymmetric (that
  is, $x_{u,v} = \neg x_{v,u}$):}
  \label{eq:2}
  x_{u,v} + x_{v,u} &= 1 && \forall u,v \in V(G), u \ne v \\
  \intertext{and transitive (that is,
  $(x_{u,v} \land x_{v,w}) \Rightarrow x_{u,w}$):}
  \label{eq:3}
  x_{u,v} + x_{v,w} - x_{u,w} &\le 1 && \forall u,v,w \in V(G),
                                        u \ne v \ne w \ne u. \\
  \intertext{Finally, we forbid monochromatic cycles and stars:}
  \label{eq:4}
  x_{u,v} + x_{v,w} + c_{u,i} + c_{w,i} &\le 3
  && \forall \{u,w\} \in E(G), v \in V(G) \setminus \{u,w\}, i \in [k]
\end{align}
In other words, if $u$ and $w$ are adjacent in~$G$ and both have
color~$i$, we disallow that there is a another vertex~$v$ between them
in the ordering.  In a SAT formulation, this can be expressed by
$\neg (x_{u,v} \land x_{v,w} \land c_{u,i} \land c_{w,i})$.  In total,
our ILP and SAT formulations have $O(n(n+k))$ variables and
$O(mnk+n^3)$ constraints, where $n=|V(G)|$ and $m=|E(G)|$.

\section{Open Problems}
\label{sec:open}

The obvious open problem is whether it is NP-hard to decide for planar
graphs of maximum vertex degree~5 whether their linear vertex
arboricity is~2.  We have done a computer search which showed that all
3-connected planar graphs with up to 12 vertices and maximum degree~5
have linear vertex arboricity at most~2.

The MSO$_2$-based algorithm in \cref{sec:fpt} runs in time linear in
the size of the given graph, but the dependency in terms of the
treewidth of the given graph and in terms of the length of the MSO$_2$
formula is very high.  For the much weaker parameter vertex cover
number, there is a simple explicit FPT algorithm for
\textsc{$k$-Lva}~\cite{erhardt-MTh21}.  It runs in
$O^\star\big(2^{c+2k {c \choose k}}\big)$ time, where $c$ is the
vertex cover number of the given graph and the $O^\star$-notations
hides polynomial factors.  What about parameters that lie between
vertex cover number and treewidth such as treedepth~-- does
\textsc{$k$-Lva} admit an explicit FPT algorithm with respect to such
a parameter?

\paragraph{Acknowledgments.}

We thank Oleg Verbitsky for reviewing the masterthesis of the first
author~\cite{erhardt-MTh21} and for suggesting several improvements in
the presentation.  We also thank a reviewer of an earlier version of
this paper for suggesting a slight strengthening of \cref{thm:fpt}.

\bibliographystyle{plainurl}
\bibliography{abbrv,segment-number}

\end{document}